\documentclass[12pt,letterpaper]{article}
\usepackage{amsmath,amsfonts,amsthm,amssymb,stmaryrd,bm,cite} 
\usepackage{relsize,tabls,mathtools} 
\usepackage{graphicx} 
\usepackage{array}   

\allowdisplaybreaks  

\theoremstyle{plain}

\numberwithin{equation}{section}
\newtheorem{thm}{Theorem}[section]
\newtheorem{lem}[thm]{Lemma}

\newcolumntype{L}{>{$}l<{$}} 
\newcolumntype{?}{!{\vrule width 1pt}} 

\newcounter{cond}

\newcommand{\integers}{{\mathbb Z}}
\newcommand{\real}{{\mathbb R}}
\newcommand{\positive}{{\mathbb N}}
\newcommand{\complex}{{\mathbb C}}
\newcommand{\cscript}{{\mathcal C}}
\newcommand{\gscript}{{\mathcal G}}
\newcommand{\hscript}{{\mathcal H}}
\newcommand{\mscript}{{\mathcal M}}
\newcommand{\sscript}{{\mathcal S}}
\newcommand{\cscripthat}{\widehat{\cscript}}
\newcommand{\sscripthat}{\widehat{\sscript}}
\newcommand{\overe}{\overline{e}}
\newcommand{\overu}{\overline{u}}
\newcommand{\overpartial}{\overline{\partial}}
\newcommand{\overgamma}{\overline{\gamma}}
\newcommand{\capahat}{\widehat{A}}
\newcommand{\capbhat}{\widehat{B}}
\newcommand{\capkhat}{\widehat{K}}
\newcommand{\capphat}{\widehat{P}}
\newcommand{\capdhat}{\widehat{D}}
\newcommand{\capfhat}{\widehat{F}}
\newcommand{\caphhat}{\widehat{H}}

\newcommand{\bfe}{\mathbf{e}}    
\newcommand{\bff}{\mathbf{f}}       
\newcommand{\bfg}{\mathbf{g}}       

\newcommand{\bfp}{\mathbf{p}}       

\newcommand{\bfx}{\mathbf{x}}       
\newcommand{\bfy}{\mathbf{y}}       
\newcommand{\bfzero}{\mathbf{0}}    
\newcommand{\bfsigma}{\mathbf{\sigma}}    
\newcommand{\ctimes}{\mathrel{\mathlarger\cdot}}
  
\newcommand{\ab}[1]{\left|#1\right|}
\newcommand{\doubleab}[1]{\left|\left|#1\right|\right|}
\newcommand{\brac}[1]{\left\{#1\right\}}
\newcommand{\paren}[1]{\left(#1\right)}
\newcommand{\sqbrac}[1]{\left[#1\right]}
\newcommand{\elbows}[1]{{\left\langle#1\right\rangle}}
\newcommand{\ket}[1]{{\left|#1\right>}}
\newcommand{\bra}[1]{{\left<#1\right|}}
\newcommand{\floordel}[1]{\lfloor#1\rfloor}  

\begin{document}
\title{DISCRETE SPACETIME\\QUANTUM FIELD THEORY
}
\author{S. Gudder\\ Department of Mathematics\\
University of Denver\\ Denver, Colorado 80208, U.S.A.\\
sgudder@du.edu
}
\date{}
\maketitle

\begin{abstract}
This paper begins with a theoretical explanation of why spacetime is discrete. The derivation shows that there exists an elementary length which is essentially Planck's length. We then show how the existence of this length affects time dilation in special relativity. We next consider the symmetry group for discrete spacetime. This symmetry group gives a discrete version of the usual Lorentz group. However, it is much simpler and is actually a discrete version of the rotation group. From the form of the symmetry group we deduce a possible explanation for the structure of elementary particle classes. Energy-momentum space is introduced and mass operators are defined. Discrete versions of the Klein-Gordon and Dirac equations are derived. The final section concerns discrete quantum field theory. Interaction Hamiltonians and scattering operators are considered. In particular, we study the scalar spin~0 and spin~1 bosons as well as the spin~$1/2$ fermion cases
\end{abstract}\newpage

\section{Why Is Spacetime Discrete?}  
Discreteness of spacetime would follow from the existence of an elementary length. Such a length, which we call a hodon \cite{cro16} would be the smallest nonzero measurable length and all measurable lengths would be integer multiplies of a hodon \cite{bdp16,cro16}. Applying dimensional analysis, Max Planck discovered a fundamental length
\begin{equation*}
\ell _p=\sqrt{\frac{\hbar G}{c^3}}\approx 1.616\times 10^{-33}\hbox{cm}
\end{equation*}
that is the only combination of the three universal physical constants $\hbar$, $G$ and $c$ with the dimension of length. However, it is of interest that $\ell _p$ can be derived using basic physical principles involving gravity, special relativity and quantum mechanics \cite{bdp16,cro16}. Although this derivation is not new, it is worth reviewing. Moreover it introduces the importance of photons for performing length and time measurements which will be needed later. One might argue that employing the above three theories, one is bound to end up with $\ell _p$ because this is the only combination of their characteristic constants with the dimension of length. But the point of this derivation is that $\ell _p$ (or something close to $\ell _p$) is the \textit{smallest} nonzero measurable length. For example, Planck also found the fundamental mass
\begin{equation*}
m_p=\sqrt{\frac{\hbar c}{G}}\approx 2.18\times 10^{-5}\hbox{gm}
\end{equation*}
and nobody has suggested this to be a smallest measurable mass. In fact, all known elementary particles have mass much smaller than $m_p$.

Probably the most accurate method of measuring lengths is by employing photons. Then the question becomes: Is there a minimal wave length for a photon? Let $E$ be an apparatus of mass $m$ that produces photons and for simplicity, we assume that $E$ has a spherical shape with radius $R$. A particle of mass $m'$ and speed $v$ on the surface of $E$ can escape the gravitational attraction of $E$ only if its kinetic energy $\frac{1}{2}m'v^2$ exceeds the gravitational potential energy $Gmm'/R$. This \textit{escape velocity} thus satisfies
\begin{equation*}
\frac{1}{2}v^2=\frac{Gm}{R}
\end{equation*}
The largest possible speed $c$ is that of a photon. We conclude that the smallest possible radius for which a photon can escape from $E$ is
\begin{equation*}
R_S=\frac{2Gm}{c^2}
\end{equation*}
The number $R_S$ (or sometimes $R_S/2$) is called the \textit{Schwarzschild radius} and in general relativity $R_S$ is taken as the radius of a black hole from which most radiation cannot escape.

What is the smallest wavelength $\lambda$ for a photon emitted by $E$? If we convert the mass $m$ of $E$ entirely into energy and assume it becomes a single photon, then by Planck's quantum postulate $mc^2=nh\nu$ where
$n\in\positive$ and $\nu$ is the photon frequency. Since we want the maximal frequency, we let $n=1$ and since
$\nu =c/\lambda$, we have that
\begin{equation*}
mc^2=\frac{hc}{\lambda}
\end{equation*}
Hence, $\lambda = h/mc$ which is called the \textit{Compton wavelength} and is denoted by $\ell _C$. This distance may be described as follows. In determining the position of a particle of mass $m$ to within a Compton wavelength, it requires enough energy to create a particle of mass $m$. If we postulate that this minimal wavelength is approximately the radius of the photon emitter $E$, we have that $\ell _C=R_S$ or in terms of the parameters
\begin{equation*}
\frac{h}{mc}=\frac{2Gm}{c^2}
\end{equation*}
Solving for $m$ gives
\begin{equation}         
\label{eq11}
m=\sqrt{\frac{hc}{2G}}
\end{equation}
The only significance of $m$ for us is when we substitute \eqref{eq11} into $\lambda =h/mc$ we obtain
\begin{equation}         
\label{eq12}
\lambda =\sqrt{\frac{2hG}{c^3}}=2\sqrt{\pi}\,\ell _p
\end{equation}
Another way of looking at this is the following. For very small masses (elementary particles) we have $R_S\ll\ell_C$ and for very large masses (stars and galaxies) we have $\ell _C\ll R_S$.. For Planck mass \eqref{eq11} which is in between, we obtain $\ell _C=R_S=\lambda$. Let us now accept the existence of an elementary length $\lambda$ called a hodon and examine the consequences. We may take $\lambda$ as given by \eqref{eq12}, but its exact value is not important for our subsequent discussion.

It follows from the existence of $\lambda$ that special relativity must break down at very small distances. This is because special relativity implies that distance measurements depend on the inertial frame in which the measurements are performed. But $\lambda$ is an absolute quantity that would be the same in all inertial frames.

We now examine length measurements more closely. Following Einstein, we assume that lengths are measured using rigid hodon length rods. It is reasonable to assume that the measured distance from point $\bfx$ to point $\bfy$ is the maximal number of hodon rods that can be placed end to end from $\bfx$ to $\bfy$. Letting $\positive _0=\brac{0,1,2,\ldots}$, if $\bfx ,\bfy\in\real ^3$, we define the \textit{discrete distance} in hodons from
$\bfx$ to $\bfy$ to be
\begin{equation*}
d(\bfx ,\bfy )=\max\brac{n\lambda\colon n\in\positive _0,n\lambda\le\doubleab{\bfx -\bfy}_3}
\end{equation*}
where $\doubleab{\bfx -\bfy}_3$ is the usual Euclidean continuum distance from $\bfx$ to $\bfy$. It follows from the definition of $d(\bfx ,\bfy )$ that
\begin{equation*}
\doubleab{\bfx -\bfy}_3<d(\bfx ,\bfy )+\lambda
\end{equation*}
Hence, $\ab{\doubleab{\bfx -\bfy}_3-d(\bfx ,\bfy )}<\lambda$ so these two distances differ by a negligible amount except for very small distance scales.

For discrete distance measurements, we have an altered Pythagoras theorem. If the two legs of a right triangle each have length one hodon, then the hypotenuse has length  one hodon. The next table summarizes such results where all the distances are the given integer multiples of a hodon.
\medskip
\begin{center}
\begin{tabular}{c ? c|c|c|c|c|c|c|c|c|c|c|c|c}
Length&\multicolumn{10}{c|}{}\\
of Legs&1&2&3&4&5&6&7&8&9&10\\
\hline
Length of&&&&&&&&&&\\
Hypotenuse&1&2&4&5&7&8&9&11&12&14\\
\hline\noalign{\medskip}
\multicolumn{11}{c}{\textbf{Table 1 (Altered Pythagoras Theorem)}}\\
\end{tabular}
\end{center}
\vskip 1pc

\noindent Again the discrete hypotenuse length differs from the continuum length by less than one hodon.

Corresponding to the elementary length $\lambda$, we have the elementary time $\lambda/c$ equal to one \textit{chronon}. Of course, one chronon is the time it takes for a photon to travel a distance of one hodon. We can measure time very accurately by using a photon clock (Einstein called them light clocks). A photon clock consists of a photon emitter-detector positioned at a point $A$ and another photon emitter-detector at point $B$. The measured discrete distance from $A$ to $B$ is $s$ hodons. We call $s$ the \textit{size} of the clock and we shall see that the clock is most accurate when $s=1$ and decreases in accuracy as $s$ gets larger. To measure time, a computer registers a ``tick'' as a photon is emitted at $A$ toward $B$. When the photon is detected at $B$, another ``tick'' (or should we say ``tock'') is registered and a new photon is emitted toward $A$. This can be repeated indefinitely. Between two ticks of the clock, a time of $s$ chronons elapses. We can determine the elapsed time for a process in multiples of a chronon by counting the number of ticks. Clearly, the clock is more accurate for smaller $s$.

Even though special relativity breaks down at distances near $\lambda$, it still holds in an altered form \cite{bdp16}. Let us consider time dilation in the discrete setting. We still follow Einstein except our rigid measuring rods are one hodon in length and our clocks measure time in chronons. Suppose a stationary observer $O_1$ possesses a photon clock of size $s$ and notes an elapsed time of $s$ chronons. Another observer $O_2$ also has a photon clock of size $s$ but $O_2$ and her clock are on a train moving at a constant speed of $v$ hodons per chronon. Of course, $v$ is a rational number with $0\le v\le 1$. Before $O_2$ boards the train, the two clocks are synchronized so that they both tick at the same time. Observer $O_2$, now on the moving train, checks her photon clock to measure her elapsed time $\Delta t'$ corresponding to $O_1$'s $\Delta t=s$. Now $O_2$'s clock has similar points $A'$, $B'$ a discrete distance of $s$ hodons apart. We assume that the line $A'B'$ is perpendicular to the direction of motion of the train. From the time that the photon is emitted from $A'$ and is detected at $B'$ in the moving train, point $A'$ has moved a horizontal distance of $x=v\Delta t'$. Letting $d$ be the total distance in hodons traveled by $O_2$'s photon during time $\Delta t'$, we have by Pythagoras' theorem that
\begin{equation*}
d^2=x^2+(\Delta t)^2=v^2(\Delta t')^2+(\Delta t)^2
\end{equation*}
Since $\Delta t'=d$ in chronons, we conclude that
\begin{equation*}
(\Delta t')^2=v^2(\Delta t')^2+(\Delta t)^2
\end{equation*}
Hence,
\begin{equation*}
\Delta t'=\frac{\Delta t}{\sqrt{1-v^2}}=\frac{s}{\sqrt{1-v^2}}
\end{equation*}

This is the same formula for time dilation as in special relativity except now $1/\sqrt{1-v^2}$ must be computed discretely so that $\Delta t'$ is an integer. We then ``round down'' and replace $1/\sqrt{1-v^2}$ by its integer part
$\floordel{(1-v^2)^{-1/2}}$ as we did for distance measurements. For $n=s,s+1,s+2,\ldots$ we have that $\Delta t'=n$ if and only if
\begin{equation}         
\label{eq13}
n\le\frac{s}{\sqrt{1-v^2}}<n+1
\end{equation}
Now \eqref{eq13} is equivalent to
\begin{equation}         
\label{eq14}
\frac{\sqrt{n^2-s^2}}{n}\le v<\frac{\sqrt{(n+1)^2-s^2}}{n+1}
\end{equation}
For example, suppose we have the most accurate possible photon clock with $s=1$. Then $\Delta t=1$ and
$\Delta t'=1$ if and only if $0\le v<\sqrt{3}/2$ or $0\le v<0.866$. In this case, there is no observed time dilation except when $v$ is quite large (that is, close to 1). 
For $\Delta t'=2$ we would have the speed $\sqrt{3}/2\le v<\sqrt{8}/3$ or $0.866\le v <0.943$. We call $\sqrt{3}/2$ the \textit{threshold speed} and $\sqrt{8}/3$ the \textit{first step speed}. We define higher steps analogously. In the high accuracy case, the threshold speed is large and the steps are relatively large. For lower accuracy clocks, in which $s$ is larger, the threshold is smaller and so are the steps. As $s\to\infty$, the continuum limit of the usual special relativity is approached. A similar analysis holds for space contraction
\begin{equation*}
\Delta x'=\Delta x\sqrt{1-v^2}
\end{equation*}

We conclude that special relativity is not only relative to the observer's motion but also relative to the accuracy of the observer's measurements. Table~2 summarizes speeds to the third step for clock sizes for $s=1,10,100$ and illustrates the approach to the smooth time dilation of special relativity as $s\to\infty$. This phenomenon may give a method for testing the discreteness of spacetime experimentally. Employing very accurate clocks moving at various speeds, one should observe whether their time dilations change at abrupt (albeit, very small) speed steps.
\medskip

\begin{center}
\begin{tabular}{|c|c|c|c|c|}
Size&Threshold&First Step&Second Step&Third Step\\
\hline
1&0.866&0.943&0.968&0.980\\
\hline
10&0.417&0.553&0.639&0.670\\
\hline
100&0.140&0.197&0.240&0.275\\
\hline\noalign{\medskip}
\multicolumn{5}{c}{\textbf{Table 2 (Size and Speeds)}}\\
\end{tabular}
\end{center}
\vskip 2pc

\section{Discrete Symmetry Group} 
The simplest form for a discrete spacetime is a cubic lattice $\sscript _4=\integers ^+\times\integers ^3$ where $\integers ^+=\brac{0,1,2,\ldots}$ represents time and
\begin{equation*}
\integers =\brac{\ldots ,-2,-1,0,1,2,\ldots}
\end{equation*}
represents a space coordinate. As before, the time units are chronons and the length units are hodons. We represent an element of $\sscript _4$ by
\begin{equation*}
x=(x_0,\bfx )=(x_0,x_1,x_2,x_3)
\end{equation*}
and employ the Minkowski distance
\begin{equation*}
\doubleab{x}_4^2=x_0^2-\doubleab{\bfx}_3^2=x_0^2-x_1^2-x_2^2-x_3^2
\end{equation*}
A \textit{symmetry} on $\sscript _4$ is a linear transformation ($4\times 4$ real matrix) $A\colon\sscript _4\to\sscript _4$ that satisfies
$\doubleab{Ax}_4^2=\doubleab{x}_4^2$ for all $x\in\sscript _4$. The set of symmetries forms a group which we denote by $\gscript _4$.

Since $A\colon\sscript_4\to\sscript _4$, the entries in $A$ must be integers. We have shown in \cite{gud16} that this forces $A$ to have the form
\begin{equation}         
\label{eq21}
A=\begin{bmatrix}\noalign{\smallskip}
1&0&0&0\\0&&&\\0&&B&\\0&&&\\\noalign{\smallskip}\end{bmatrix}
\end{equation}
where $B\colon\integers ^3\to\integers ^3$ is an orthogonal $3\times 3$ matrix. We now motivate why this happens with a lower dimensional example. Consider a 2-dimensional discrete spacetime $\sscript _2=\integers ^+\times\integers$ with $\doubleab{x}_2^2=x_0^2-x_1^2$. It is not hard to show that if $A\colon\real ^2\to\real ^2$ is a real matrix such that $\doubleab{Ax}_2^2=\doubleab{x}_2^2$ for all $x\in\real ^2$, then $A$ has the form
\begin{equation}         
\label{eq22}
A=\begin{bmatrix}\noalign{\smallskip}
\pm\sqrt{1+a^2}&a\\\noalign{\smallskip}a&\pm\sqrt{1+a^2}\\\noalign{\smallskip}\end{bmatrix}
\end{equation}
for some $a\in\real$. If $a\ne  0 $, this is called a \textit{velocity boost}. We can check this by applying $A$ given in \eqref{eq22} with plus signs to obtain
\begin{equation*}
A=\begin{bmatrix}
c\\\noalign{\smallskip}d\end{bmatrix}
=\begin{bmatrix}\noalign{\smallskip}
\sqrt{1+a^2}\,c+ad\\\noalign{\smallskip}ac+\sqrt{1+a^2}\,d\\\noalign{\smallskip}\end{bmatrix}
=\begin{bmatrix}
c'\\\noalign{\smallskip}d'\end{bmatrix}
\end{equation*}
Then
\begin{align*}
\doubleab{\,\begin{bmatrix}c'\\\noalign{\smallskip}d'\end{bmatrix}\,}_2^2
&=\paren{\sqrt{1+a^2}\,c+ad}^2-\paren{ac+\sqrt{1+a^2}\,d}^2\\
&=c^2-d^2=\doubleab{\,\begin{bmatrix}c\\\noalign{\smallskip}d\end{bmatrix}\,}_2^2
\end{align*}
Now if $A\colon\sscript _2\to\sscript _2$, then $A$ must have integers entries. But the only $a\in\integers$ with $\sqrt{1+a^2}\in\integers$ is $a=0$. We conclude that the symmetry group on $\sscript _2$ consists of only two elements
\begin{equation*}
\begin{bmatrix}
1&0\\\noalign{\smallskip}0&1\end{bmatrix},\quad
\begin{bmatrix}1&0\\\noalign{\smallskip}0&-1\end{bmatrix}
\end{equation*}
We shall show later that $\gscript _4$ has order 48.

We denote the group of orthogonal $3\time 3$ matrices $B\colon\integers ^3\to\integers ^3$ by $O_3$ so that
$\gscript _4=1\times O_3$ where 1 is the 1-dimensional identity. It is also useful to consider the group of discrete rotations $O'_3$ that consists of the elements of $O_3$ with determinant 1. A simple element of $O'_3$ is the rotation of $\pi/2$ about the $x_1$ axis. This rotation has matrix
\begin{equation*}
A=\begin{bmatrix}
1&0&0\\\noalign{\smallskip}0&0&-1\\\noalign{\smallskip}0&1&0\end{bmatrix}
\end{equation*}
We also have rotations of $\pi/2$ about the $x_2$ and $x_3$ axes. These are given by
\begin{equation*}
B=\begin{bmatrix}
0&0&1\\\noalign{\smallskip}0&1&0\\\noalign{\smallskip}-1&0&0\end{bmatrix},\quad
C=\begin{bmatrix}
0&-1&0\\\noalign{\smallskip}1&0&0\\\noalign{\smallskip}0&0&1\end{bmatrix}\end{equation*}
Any two of these three matrices generate the entire discrete rotation group $O'_3$. This group has order 24 and its elements besides $A,B,C$ are presented below. The group table for $O'_3$ is given in Table~3.

\newpage
\begin{align*}
&\hskip -10pt D=\!\begin{bmatrix}\noalign{\smallskip}
1&0&0\\0&-1&0\\0&0&-1\\\noalign{\smallskip}\end{bmatrix}\!,\ 
E=\!\begin{bmatrix}\noalign{\smallskip}
1&0&0\\0&0&1\\0&-1&0\\\noalign{\smallskip}\end{bmatrix}\!,\ 
F=\!\begin{bmatrix}\noalign{\smallskip}
-1&0&0\\0&1&0\\0&0&-1\\\noalign{\smallskip}\end{bmatrix}\!,\ 
G=\!\begin{bmatrix}\noalign{\smallskip}
0&0&-1\\0&1&0\\1&0&0\\\noalign{\smallskip}\end{bmatrix}\\\noalign{\medskip}
&\hskip -10pt H=\begin{bmatrix}\noalign{\smallskip}
-1&0&0\\0&-1&0\\0&0&1\\\noalign{\smallskip}\end{bmatrix},\ 
I=\begin{bmatrix}\noalign{\smallskip}
1&0&0\\0&1&0\\0&0&1\\\noalign{\smallskip}\end{bmatrix},\ 
J=\begin{bmatrix}\noalign{\smallskip}
0&1&0\\-1&0&0\\0&0&1\\\noalign{\smallskip}\end{bmatrix},\ 
K=\begin{bmatrix}\noalign{\smallskip}
0&0&1\\1&0&0\\0&1&0\\\noalign{\smallskip}\end{bmatrix}\\\noalign{\medskip}
&\hskip -10pt L=\begin{bmatrix}\noalign{\smallskip}
0&1&0\\0&0&-1\\-1&0&0\\\noalign{\smallskip}\end{bmatrix}\!,\ 
M\!=\!\begin{bmatrix}\noalign{\smallskip}
0&-1&0\\0&0&-1\\1&0&0\\\noalign{\smallskip}\end{bmatrix}\!,\ 
N\!=\!\begin{bmatrix}\noalign{\smallskip}
0&-1&0\\0&0&1\\-1&0&0\\\noalign{\smallskip}\end{bmatrix}\!,\ 
O\!=\!\begin{bmatrix}\noalign{\smallskip}
-1&0&0\\0&0&1\\0&1&0\\\noalign{\smallskip}\end{bmatrix}\\\noalign{\medskip}
&\hskip -10pt P\!=\!\begin{bmatrix}\noalign{\smallskip}
-1&0&0\\0&0&-1\\0&-1&0\\\noalign{\smallskip}\end{bmatrix},\ 
Q\!=\!\begin{bmatrix}\noalign{\smallskip}
0&0&-1\\-1&0&0\\0&1&0\\\noalign{\smallskip}\end{bmatrix},\ 
R\!=\!\begin{bmatrix}\noalign{\smallskip}
0&0&1\\0&-1&0\\1&0&0\\\noalign{\smallskip}\end{bmatrix},\ 
S\!=\!\begin{bmatrix}\noalign{\smallskip}
0&1&0\\1&0&0\\0&0&-1\\\noalign{\smallskip}\end{bmatrix}\\\noalign{\medskip}
&\hskip -10pt T\!=\!\!\begin{bmatrix}\noalign{\smallskip}
0&0&-1\\0&-1&0\\-1&0&0\\\noalign{\smallskip}\end{bmatrix}\!,\ 
U\!=\!\begin{bmatrix}\noalign{\smallskip}
0&-1&0\\-1&0&0\\0&0&-1\\\noalign{\smallskip}\end{bmatrix}\!,\ 
V\!=\!\begin{bmatrix}\noalign{\smallskip}
0&0&1\\-1&0&0\\0&-1&0\\\noalign{\smallskip}\end{bmatrix}\!,\ 
W\!=\!\!\begin{bmatrix}\noalign{\smallskip}
0&1&0\\0&0&1\\1&0&0\\\noalign{\smallskip}\end{bmatrix}\\\noalign{\medskip}
&\hskip -10pt X\!=\!\begin{bmatrix}\noalign{\smallskip}
0&0&-1\\1&0&0\\0&-1&0\\\noalign{\smallskip}\end{bmatrix}\\
\end{align*}
\newpage

{\parindent=-75pt
\begin{tabular}{|L|L|L|L|L|L|L|L|L|L|L|L|L|L|L|L|L|L|L|L|L|L|L|L|L|L|L|L|}
&I&A&B&C&D&E&F&G&H&J&K&L&M&N&O&P&Q&R&S&T&U&V&W&X\\
\hline
I&I&A&B&C&D&E&F&G&H&J&K&L&M&N&O&P&Q&R&S&T&U&V&W&X\\
\hline
A&A&D&K&M&E&I&O&Q&P&L&R&S&U&C&H&F&T&V&W&X&N&B&J&G\\
\hline
B&B&L&F&K&T&N&G&I&R&V&S&P&A&O&W&M&J&D&X&H&Q&U&E&C\\
\hline
C&C&K&N&H&S&X&U&M&J&I&O&B&R&T&Q&V&A&W&F&L&D&E&G&P\\
\hline
D&D&E&R&U&I&A&H&T&F&S&V&W&N&M&P&O&X&B&J&G&C&K&L&Q\\
\hline
E&E&I&V&N&A&D&P&X&O&W&B&J&C&U&F&H&G&K&L&Q&M&R&S&T\\
\hline
F&F&P&G&S&H&O&I&B&D&U&X&M&L&W&E&A&V&T&C&R&J&Q&N&K\\
\hline
G&G&M&I&X&R&W&B&F&T&Q&C&A&P&E&N&L&U&D&K&D&V&J&O&S\\
\hline
H&H&O&T&J&F&P&D&R&I&C&Q&N&W&L&A&E&K&G&U&B&S&X&M&V\\
\hline
J&J&Q&L&I&U&V&S&W&C&H&A&T&G&B&K&X&O&M&D&N&F&P&R&E\\
\hline
K&K&S&O&R&X&C&Q&A&B&V&W&F&D&H&J&U&L&E&G&P&T&N&I&M\\
\hline
L&L&T&S&A&N&B&W&J&M&P&D&X&Q&K&R&G&H&U&E&C&O&F&V&I\\
\hline
M&M&R&C&P&W&G&N&U&L&A&H&K&V&X&T&B&D&J&O&S&E&I&Q&F\\
\hline
N&N&B&U&O&L&T&M&C&W&E&F&V&K&Q&G&R&I&S&P&J&A&D&X&H\\
\hline
O&O&F&Q&W&P&H&A&K&E&N&G&U&S&J&I&D&B&X&M&V&L&T&C&R\\
\hline
P&P&H&X&L&O&F&E&V&A&M&T&C&J&S&D&I&R&Q&N&K&W&G&U&B\\
\hline
Q&Q&U&A&G&V&J&K&O&X&T&M&D&F&I&C&S&N&P&R&E&B&L&H&W\\
\hline
R&R&W&H&V&G&M&T&D&B&K&J&O&E&P&L&N&S&I&Q&F&X&C&A&U\\
\hline
S&S&X&W&D&C&K&J&L&U&F&E&G&T&R&V&Q&P&N&I&M&H&O&B&A\\
\hline
T&T&N&D&Q&B&L&R&H&G&X&U&E&O&A&M&W&C&F&V&I&K&S&P&J\\
\hline
U&U&V&M&F&J&Q&C&N&S&D&P&R&B&G&X&K&E&L&H&W&I&A&T&O\\
\hline
V&V&J&P&B&Q&U&X&E&K&R&L&H&I&F&S&C&W&A&T&O&G&M&D&N\\
\hline
W&W&G&J&E&M&R&L&S&N&O&I&Q&X&V&B&T&F&C&A&U&P&H&K&D\\
\hline
X&X&C&E&T&K&S&V&P&Q&G&N&I&H&D&U&J&M&O&B&A&R&W&F&L\\
\hline\noalign{\medskip}
\multicolumn{24}{c}{\textbf{Table 3 (The Group $O'_3$)}}\\
\end{tabular}
\parindent=18pt}
\vskip 3pc

It is not hard to show that the full discrete orthogonal group $O_3$ has the form $O_3=O'_3\cup (-O'_3)$. Where
\begin{equation*}
-O'_3=\brac{-Z\colon Z\in O'_3}
\end{equation*}
are the elements of $O_3$ with determinant $-1$. It follows that $O'_3$ is a subgroup of $O_3$ and that the order of $O_3$ is 48.

\begin{lem}       
\label{lem21}
If $\gscript$ is a subgroup of $O'_3$, then $\gscript\cup (-\gscript )$ is a subgroup of $O_3$.
\end{lem}
\begin{proof}
If $Y,Z\in\gscript$, then $YZ$ and $(-Y)(-Z)$ are in $\gscript$. Also, $Y(-Z)=-YZ$ and $(-Y)Z=-YZ$ are in $-\gscript$. Thus, $\gscript\cup (-\gscript )$ is closed under the group product. Moreover, if $Y\in\gscript$, then $Y^{-1}\in\gscript$ and if $-Y\in -\gscript$ then $(-Y)^{-1}=-Y^{-1}\in -\gscript$. Hence, $\gscript\cup(-\gscript )$ is closed under inverses so
$\gscript\cup (-\gscript )$ is a group.
\end{proof}

Not all subgroups $\gscript$ of $O_3$ have the form $\gscript\subseteq O'_3$ or $G=\hscript\cup (-\hscript )$ where $\hscript$ is a subgroup of $O'_3$. For example, since $P^2=I$ we have that $\brac{I,-P}$ is a subgroup of $O_3$ which is not of the above form.

We now examine the eigenvalue-eigenvector structure of the elements of $O'_3$. The eigenvalues of $A$ are $1,i,-i$ with corresponding normalized eigenvectors
\begin{equation*}
\begin{bmatrix}1\\0\\0\end{bmatrix},\quad
\frac{1}{\sqrt{2}}\begin{bmatrix}0\\1\\-i\end{bmatrix},\quad
\frac{1}{\sqrt{2}}\begin{bmatrix}0\\1\\i\end{bmatrix}
\end{equation*}
respectively. The projections onto these eigenvectors are
\begin{equation*}
P(1)=\begin{bmatrix}
1&0&0\\\noalign{\smallskip}0&0&0\\\noalign{\smallskip}0&0&0\end{bmatrix},\quad
P(i)=\frac{1}{2}\begin{bmatrix}0&0&0\\\noalign{\smallskip}0&1&i\\\noalign{\smallskip}0&-i&1\end{bmatrix},\quad
P(-i)=\frac{1}{2}\begin{bmatrix}0&0&0\\\noalign{\smallskip}0&1&-i\\\noalign{\smallskip}0&i&1\end{bmatrix}
\end{equation*}
respectively. Now $A$ has a self-adjoint \textit{momentum operator} $\capahat$ satisfying\newline
$A=e^{i\capahat}$. It follows that $\capahat =-i\ln A$ so $\capahat$ has eigenvalues $0,\pi/2,-\pi /2$. We conclude that
\begin{equation*}
\capahat =0P(1)+\frac{\pi}{2}\,P(i)-\frac{\pi}{2}\,P(-i)=i\,\frac{\pi}{2}
\begin{bmatrix}0&0&0\\\noalign{\smallskip}0&0&1\\\noalign{\smallskip}0&-1&0\end{bmatrix}
\end{equation*}
All the elements of the \textit{lepton set} $\brac{A,B,C,E,G,J}\subseteq O'_3$ have similar eigenvalue-eigenvector structure (we explain this terminology later). In particular, they have the eigenvalues $1,i,-i$ with slightly different eigenvectors. For example, the corresponding eigenvectors for $B$ are
\begin{align*}
\begin{bmatrix}0\\1\\0\end{bmatrix},\quad
\frac{1}{\sqrt{2}}\begin{bmatrix}-i\\0\\1\end{bmatrix}.\quad
\frac{1}{\sqrt{2}}\begin{bmatrix}i\\0\\1\end{bmatrix}
\intertext{and}
\capbhat =i\,\frac{\pi}{2}
\begin{bmatrix}0&0&-1\\\noalign{\smallskip}0&0&0\\\noalign{\smallskip}1&0&0\end{bmatrix}
\end{align*}

The eigenvalues of $K$ are $1,e^{i2\pi/3},e^{-i2\pi /3}$ and the corresponding eigenvectors are
\begin{equation*}
\frac{1}{\sqrt{3}}\begin{bmatrix}1\\1\\1\end{bmatrix},\quad
\frac{1}{\sqrt{3}}\begin{bmatrix}1\\e^{-i2\pi /3}\\e^{i2\pi /3}\end{bmatrix}.\quad
\frac{1}{\sqrt{3}}\begin{bmatrix}1\\e^{i2\pi /3}\\e^{-i2\pi /3}\end{bmatrix}
\end{equation*}
The momentum operator has eigenvalues $0,-\frac{2\pi}{3},\frac{2\pi}{3}$ and we have
\begin{equation*}
\capkhat =\frac{2\pi}{9}
\begin{bmatrix}0&1&1\\\noalign{\smallskip}1&0&1\\\noalign{\smallskip}1&1&0\end{bmatrix}
\end{equation*}
Again the eigenvalues for the elements of the \textit{gluon set} $\brac{K,L,M,N,Q,V,W,X}$ are the same as for $K$ and their eigenvectors and momentum operators are similar.

The eigenvalues of $P$ are $1,-1,-1$ with corresponding eigenvectors
\begin{equation*}
\frac{1}{\sqrt{2}}\begin{bmatrix}0\\1\\-1\end{bmatrix},\quad
\frac{1}{\sqrt{3}}\begin{bmatrix}1\\1\\1\end{bmatrix}.\quad
\frac{1}{\sqrt{6}}\begin{bmatrix}-2\\1\\1\end{bmatrix}
\end{equation*}
The momentum operator has eigenvalues $0,\pi ,\pi$ and we have
\begin{equation*}
\capphat =\frac{\pi}{2}
\begin{bmatrix}2&0&0\\\noalign{\smallskip}0&1&1\\\noalign{\smallskip}0&1&1\end{bmatrix}
\end{equation*}
As before, the eigenvalues are the same for all elements of the \textit{quark set} $\brac{O,P,R,S,T,U}$ and their eigenvectors and momentum operators are similar.

We are left with the \textit{boson set} $\brac{D,F,H}$. These are the only matrices in $O'_3$ that are diagonal
(except for I). By inspection, they have eigenvalues $1,-1,-1$ and eigenvectors
\begin{equation*}
\begin{bmatrix}1\\0\\0\end{bmatrix},\quad
\begin{bmatrix}0\\1\\0\end{bmatrix}.\quad
\begin{bmatrix}0\\0\\1\end{bmatrix}
\end{equation*}
is some order. The momentum operator for $D$ is
\begin{equation*}
\capdhat =\pi
\begin{bmatrix}0&0&0\\\noalign{\smallskip}0&1&0\\\noalign{\smallskip}0&0&1\end{bmatrix}
\end{equation*}
and $\capfhat$, $\caphhat$ are similar.

The group $O'_3$ can be represented as the physical symmetries of a cube which we denote by $Q'_3$. We can describe the elements of $Q'_3$ by placing a spindle through the center of the cube. There are three possible ways for a spindle to give a symmetry.
\begin{list} {\arabic{cond}.}{\usecounter{cond}
\setlength{\rightmargin}{\leftmargin}}
\item The spindle pierces the cube at the center of two opposite faces. There are three pairs of opposite faces and three nontrivial rotations by $\pi /2$, $\pi$ and $3\pi /2$. This gives 9 symmetries.
\item The spindle pierces the center between two opposite edges. There are 12 edges and six pairs of opposite edges. There is one nontrivial rotation by $\pi$. This gives six symmetries.
\item The spindle pierces two opposite corners (vertices). There are 8 vertices and four pairs of opposite vertices. There are two nontrivial rotations by $2\pi /3$ and $4\pi /3$. This gives 8 symmetries.
\end{list}

Including the identity, we obtain a total of 24 physical symmetries of a cube. Figure~1 illustrates a cube with its vertices labeled.
\begin{center}
\includegraphics[scale=.75]{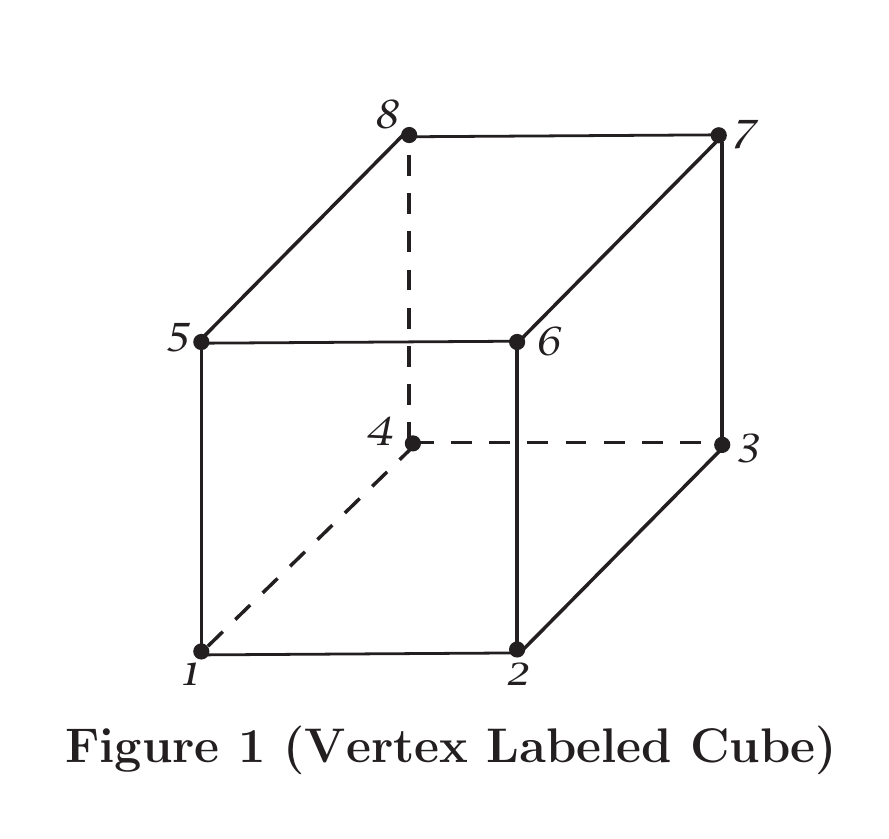}
\end{center}
We can describe the symmetries of a cube in terms of the permutations of its vertices. These are given as follows.
\begin{align*}
a&=(1265)(3784),\quad b=(1485)(2376),\quad c=(1234)(5678)\\
d&=(16)(25)(38)(47),\quad e=(1562)(34787),\quad f=(18)(27)(36)(45)\\
g&=(1584)(2673),\quad h=(13)(24)(57)(68),\quad i=(1)(2)(3)(4)(5)(6)(7)(8)\\
j&=(1432)(5876),\quad k=(138)(275)(4)(6),\quad l=(136)(475)(2)(8)\\
m&=(168)(274)(3)(5),\quad n=(245)(638)(1)(7),\quad o=(14)(28)(35)(67)\\
p&=(17)(23)(46)(58),\quad q=(254)(368)(1)(7),\quad r=(17)(28)(34)(56)\\
s&=(17)(26)(35)(48),\quad t=(12)(35)(46)(78),\quad u=(15)(28)(33)(46)\\
v&=(186)(247)(3)(5),\quad w=(183)(257)(4)(6),\quad x=(163)(457)(2)(8)
\end{align*}
The groups $O'_3$ and $Q'_3$ are isomorphic with isomorphism given by $A\to a,B\to b,\ldots ,X\to x$. We can divide the elements of $Q'_3$ into three types according to their spindle classification given above.

We see that $a,b,c,e,g,j$ each have order four and they are type~1. Also, $d,f,h$ each have order two and they are type~1. The symmetries $k,l,m,n,q$, $v,w,x$ have order three and type~3. Finally, $o,p,r,s,t,u$ have order two and type~2. This classification of similar elements of $Q'_3$ and the standard classification of elementary particles seems more than a coincidence. We suggest that $\brac{a,b,c,e,g,j}$ correspond to the six leptons; $\brac{o,p,r,s,t,u}$ correspond to the six quarks; $\brac{k,l,m,n,q,v,w,x}$ correspond to the eight gluons and $\brac{d,f,h}$ correspond to the two weak bosons and the Higgs boson. The anti-particles of these particles are given by the rest of the group elements in $Q_3$ that we now discuss.

We have viewed the group $O'_3$ in two ways. One way as a group of $3\times 3$ orthogonal matrices with determinant 1 and the other as the group $Q'_3$ of physical rotations of a cube. Of course, the two views are equivalent and the groups are isomorphic. But then we extended $O'_3$ to the group $O_3=O'_3\cup\brac{-Q'_3}$

\section{\!\!Discrete Spacetime and Energy-Momentum} 
Our basic assumption is that spacetime is discrete and has the form of a 4-dimensional cubic lattice $\sscript _4=\positive ^+\times\positive ^3$ \cite{gud16,gud17}. We regard $\sscript _4$ as a framework or scaffolding in which vertices represent tiny cells of Planck scale that may or may not be occupied by a particle. The edges connecting vertices represent directions in which particles can propagate. Note that $\sscript _4$ is a module in the sense that $\sscript _4$ is closed under addition and multiplication by elements of $\positive ^+$. The vectors
\begin{align*}
d&=(1,\bfzero )=(1,0,0,0)\\
e&=(0,\bfe )=(0,1,0,0)\\
f&=(0,\bff )=(0,0,1,0)\\
g&=(0,\bfg )=(0,0,0,1)
\end{align*}
form a basis for $\sscript _4$ and every element of $\sscript _4$ has the unique form
\begin{equation*}
x=nd+me+pf+qg
\end{equation*}
$n\in\positive ^+$, $m,p,q\in\positive$. As usual, $x,y\in\sscript _4$ are \textit{time-like separated} if
$\doubleab{x-y}_4^2\ge 0$ and \textit{space-like separated} if $\doubleab{x-y}_4^2<0$ and we are using units in which the speed of light is 1.

The dual of $\sscript _4$ is denoted by $\sscripthat _4$. We regard $\sscripthat _4$ as having the identical structure as
$\sscript _4$ and that $\sscripthat _4$ again has basis $d,e,f,g$. The only difference is that we denote elements of
$\sscripthat _4$ by 
\begin{equation*}
p=(p_0,\bfp)=(p_0,p_1,p_2,p_3)
\end{equation*}
and interpret $p$ as the energy-momentum vector for a particle. In fact, we sometimes call $p\in\sscripthat _4$ a particle. Moreover, we only consider the forward cone
\begin{equation*}
\cscripthat ^+(0)=\brac{p\in\sscripthat _4\colon\doubleab{p}_4\ge 0}
\end{equation*}
in $\sscripthat _4$. For a particle $p\in\sscripthat _4$ we call $p_0\ge 0$ the \textit{total energy},
$\doubleab{\bfp}_3\ge 0$ the \textit{kinetic energy} and $m=\doubleab{p}_4\ge 0$ the \textit{mass} of $p$. The integers $p_1,p_2,p_3$ are the \textit{momentum components} of $p$. Since
\begin{equation*}
m^2=\doubleab{p}_4^2=p_0^2-\doubleab{\bfp}_3^2
\end{equation*}
we conclude that Einstein's energy formula $p_0^2=m^2+\doubleab{\bfp}_3^2$ holds. The set of particles with mass $m$ is represented by the \textit{mass hyperboloid}
\begin{equation*}
\Gamma _m=\brac{p\in\sscripthat _4\colon\doubleab{p}_4=m}
\end{equation*}
In particular, $\Gamma _0$ gives the set of zero mass particles which we take to be the photons.

We define the \textit{velocity} of a particle $p=(p_0,\bfp )$ with $p_0\ne 0$, to be
\begin{equation*}
\frac{\bfp}{p_0}=\frac{\bfp}{\sqrt{m^2+\doubleab{\bfp}_3^2}}
\end{equation*}
and its \textit{speed} to be $\doubleab{\bfp}_3/p_0$. In particular, the velocity of a photon is $\bfp /\doubleab{\bfp}_3$ and its speed is 1. For example, the photon $(3,2\bfe +2\bff +\bfg )$ has velocity $(2\bfe +2\bff +\bfg )/3$. As another example, the particle $(2,\bfe +\bff )$ has velocity $(\bfe +\bff )/2$ and speed $1/\sqrt{2}$. Notice that the speed $s$ can be written as
\begin{equation*}
s=\frac{\sqrt{p_0^2-m^2}}{p_0}=\sqrt{1-\frac{m^2}{p_0^2}}
\end{equation*}
Of course, $0\le s\le 1$ and $s=1$ if and only if $m=0$. Moreover, $s=0$ if and only if $\doubleab{\bfp}_3=0$. Photons with even moderate energy can have momentum vectors pointing in millions of directions. This is why spacetime appears to be isotropic in all directions.

As discussed in Section~2, the isometries on $\sscript _4$ and $\sscripthat _4$ have the form $\gscript _4=1\times O_3$ where $O_3$ is the set of orthogonal matrices on $\positive _3$. It follows that the indefinite inner product
\begin{equation*}
px=p_0x_0-p_1x_1-p_2x_2-p_3x_3
\end{equation*}
is invariant under $\gscript _4$ for all $p\in\sscripthat _4$, $x\in\sscript _4$. That is $ApAx =px$ for all $A\in\gscript _4$. In the next section an important role will be played by the complex-valued function $e^{i\pi px/2}$. Since $px$ has integer values, it follows that $e^{i\pi px/2}$ only has the four values $\pm 1,\pm i$. Although we cannot define derivatives in this discrete framework, we can still define the operators $\partial\mu$, $\mu =0,1,2,3$, by
\begin{equation*}
\partial _0e^{i\pi px/2}=p_0e^{i\pi px/2},\quad \partial _je^{i\pi px/2}=-p_je^{i\pi px/2},\quad j=1,2,3
\end{equation*}
Similarly, we define
\begin{equation*}
\partial _0e^{-i\pi px/2}=-p_oe^{i\pi px/2},\quad \partial _je^{-i\pi px/2}=p_je^{-i\pi px/2},\quad j=1,2,3
\end{equation*}
Thus, $\partial _\mu e^{i\pi px/2}$ and $\partial _\mu e^{-i\pi px/2}$ have only the values $\pm p_\mu ,\pm ip _{\mu '}$.
We assume that the operators $\partial _\mu$ are linear. Another important discrete operator is
\begin{equation*}
\square =\partial _0^2-\partial _1^2-\partial _2^2-\partial _3^2
\end{equation*}
If $p\in\Gamma _m$ we have a simple version of the Klein-Gordon equation
\begin{equation*}
\square\,e^{\pm i\pi px/2}=m^2e^{\pm i\pi px/2}
\end{equation*}

The \textit{Pauli matrices} are given by
\begin{equation*}
\sigma _0=\begin{bmatrix}1&0\\0&1\end{bmatrix},\quad
\sigma _1=\begin{bmatrix}0&1\\1&0\end{bmatrix},\quad
\sigma _2=\begin{bmatrix}0&-i\\i&0\end{bmatrix},\quad
\sigma _3=\begin{bmatrix}1&0\\0&-1\end{bmatrix}
\end{equation*}
These matrices are self-adjoint and $\sigma _1,\sigma _2,\sigma _3$ have eigenvalues $\pm 1$ with corresponding eigenvectors
\begin{equation*}
\frac{1}{\sqrt{2}}\begin{bmatrix}1\\1\end{bmatrix},\ \frac{1}{\sqrt{2}}\begin{bmatrix}1\\-1\end{bmatrix};\quad
\frac{1}{\sqrt{2}}\begin{bmatrix}1\\i\end{bmatrix},\ \frac{1}{\sqrt{2}}\begin{bmatrix}1\\-i\end{bmatrix};\quad
\begin{bmatrix}1\\0\end{bmatrix},\ \begin{bmatrix}0\\1\end{bmatrix}
\end{equation*}
We define the 2-\textit{dimensional mass operator} by
\begin{align*}
p\sigma&=p_0\sigma _0-\bfsigma\ctimes\bfp =p_0\sigma _0-p_1\sigma _1-p_2\sigma _2-p_3\sigma _3\\
  &=\begin{bmatrix}p_0-p_3&-p_1+ip_2\\-p_1-ip_2&p_0+p_3\end{bmatrix}
\end{align*}
We have that $p\sigma$ is self-adjoint with eigenvalues $p_0\pm\doubleab{\bfp}_3$, determinant $m^2$ when
$p\in\Gamma _m$ and corresponding eigenvectors
\begin{equation*}
\sqrt{\frac{\doubleab{\bfp}_3\pm p_0}{2\doubleab{\bfp}_3}}
\begin{bmatrix}1\\\frac{p_1+ip_2}{\pm\doubleab{\bfp}_3+p_0}\end{bmatrix}
\end{equation*}

This mass operator is not entirely satisfactory because it's eigenvalues are not directly related to the mass and it is only the determinant that gives $m^2$. This operator was originally intended to describe spin $1/2$ particles. We now discuss a more satisfactory mass operator which is 4-dimensional. It was originally introduced by Dirac to describe spin $1/2$ particles and anti-particles. W define the $4\times 4$ $\gamma$-matrices by
\begin{equation*}
\gamma _0=\begin{bmatrix}\sigma _0&0\\0&-\sigma _0\end{bmatrix},\quad
\gamma _j=\begin{bmatrix}0&\sigma _j\\-\sigma _j&0\end{bmatrix},\ \quad j=1,2,3
\end{equation*}
The 4-\textit{dimensional mass operator} is defined as
\begin{align*}
\mscript _p&=p\gamma =p_0\gamma _0-p_1\gamma _1-p_2\gamma _2-p_3\gamma _3
  =\begin{bmatrix}p_0\sigma _0&-\bfp\ctimes\bfsigma\\\bfp\ctimes\bfsigma&-p_0\sigma _0\end{bmatrix}\\
  \noalign{\smallskip}
  &=\begin{bmatrix}p_0&0&-p_3&-p_1+ip_2\\ 0&p_0&-p_1-ip_2&p_3\\
  p_3&p_1-ip_2&-p_0&0\\ p_1+ip_2&-p_2&0&-p_0\end{bmatrix}
\end{align*}
Since $\mscript _p$ described particles and antiparticles, it is not self-adjoint. The reason for the minus signs in the
$\gamma$-matrices is that the antiparticles are space inversion (and negation of charges) copies of particles. The space inversion is also evident in our group theoretic description $-A,A\in O'_3$, for antiparticles. Notice that
\begin{align*}
(\bfp\ctimes\bfsigma )^2
  &=\begin{bmatrix}p_3&p_1-ip_2\\p_1+ip_2&-p_3\end{bmatrix}^2=\doubleab{\bfp}_3^2\sigma _0\\
\intertext{and hence,}
\mscript _p^2&=\begin{bmatrix}p_0\sigma _0&-\bfp\ctimes\bfsigma\\\bfp\ctimes\bfsigma&-p_0\sigma _0\end{bmatrix}^2
=\begin{bmatrix}p_0^2\sigma _0-(\bfp\ctimes\bfsigma )^2&0\\0&p_0^2\sigma _0-(\bfp\ctimes\bfsigma )^2\end{bmatrix}\\
  \noalign{\medskip}
  &=\begin{bmatrix}\paren{p_0^2-\doubleab{\bfp}_3^2}\sigma _0&0\\
  0&\paren{p_0^2-\doubleab{\bfp}_3^2}\sigma _0\end{bmatrix}
\end{align*}

We conclude that if $p\in\Gamma _m$ then $\mscript _p^2=m^2I$. This indicates that the eigenvalues of $\mscript _p$ are $\pm m$ when $p\in\Gamma _m$. This is, in fact the case and the eigenvectors corresponding to eigenvalue $m$ are
\begin{equation}         
\label{eq31}
u_1=\sqrt{\frac{p_0+m}{2p_0}}\begin{bmatrix}1\\0\\\frac{p_3}{p_0+m}\\\noalign{\medskip}
  \frac{p_1+ip_2}{p_0+m}\\\noalign{\smallskip}\end{bmatrix},\qquad
u_2=\sqrt{\frac{p_0+m}{2p_0}}\begin{bmatrix}0\\1\\\frac{p_1-ip_2}{p_0+m}\\\noalign{\medskip}
  \frac{-p_3}{p_0+m}\\\noalign{\smallskip}\end{bmatrix}
\end{equation}
These eigenvectors correspond to spin $1/2$ particles with spin up and spin down, respectively. Notice that
$\elbows{u_1,u_2}=0$. The eigenvectors corresponding to eigenvalue $-m$ are
\begin{equation}         
\label{eq32}
u_3=\sqrt{\frac{p_0+m}{2p_0}}\begin{bmatrix}\noalign{\smallskip}\frac{p_3}{p_0+m}\\\noalign{\medskip}
  \frac{p_1+ip_2}{p_0+m}\\1\\0\end{bmatrix},\qquad
u_4=\sqrt{\frac{p_0+m}{2p_0}}\begin{bmatrix}\noalign{\smallskip}\frac{p_1-ip_2}{p_0+m}\\\noalign{\medskip}
  \frac{-p_3}{p_0+m}\\0\\1\end{bmatrix}
\end{equation}
These eigenvectors correspond to spin $1/2$ antiparticles with spin up and down, respectively. Again, we have that
$\elbows{u_3,u_4}=0$. Notice that for a particle at rest when $\bfp =0$ we have that $\mscript _p=m\gamma _0$ and the eigenvectors become $u_1=d$, $u_2=e$, $u_3=f$, $u_4=g$. In general, the equations $\mscript _pu=\pm mu$ are discrete versions of Dirac's equations.

\section{Discrete Quantum Field Theory} 
This section treats discrete quantum field theory for spin 0 and spin 1 bosons and fermions with spin $1/2$. These cover most of the important cases and others can be treated in a similar way \cite{ps95,vel94}. We begin with the simplest case which is a spin 0 boson with mass $m$. Let $K$ be a complex Hilbert space whose unit vectors are the states of the system. These states represent the possible superpositions of various combinations of mass $m$, spin 0 bosons each having a characteristic energy-momentum $p\in\Gamma _m$. The basic operators are the annihilation and creation operators $a(p)$, $a(p)^*$, $p\in\Gamma _m$. The operator $a(p)$ represents the annihilation of a boson with
$p\in\Gamma _m$ and its adjoint $a(p)^*$ the creation of a boson with $p\in\Gamma _m$. These operators satisfy the commutation relations \cite{ps95,vel94}
\begin{align*}
\sqbrac{a(p),a(q)^*}&=a(p)a(q)^*-a(q)^*a(p)=\partial _{pq}I\\
\sqbrac{a(p),a(q)}&=\sqbrac{a(p)^*,a(q)^*}=0
\end{align*}
A \textit{free boson quantum field} is a map $\gamma$ from $\sscript _4$ into the set of (unbounded) operators on $K$ of the form
\begin{equation}         
\label{eq41}
\gamma (x)=\sum _{p\in\Gamma _m}\frac{1}{p_0}\sqbrac{f(p)a(p)e^{i\pi px/2}+g(p)a(p)^*e^{-i\pi px/2}}
\end{equation}
where $f,g\colon\Gamma _m\to\complex$. The most common example is when $f(p)=g(p)=1$ for all $p\in\Gamma _m$. Notice that $\gamma (x)$ is self-adjoint when $g(p)=\overline{f(p)}$ for all $p\in\Gamma _m$

If $\gamma$ is given by \eqref{eq41} and $A\in\gscript _4$ we have the boson field
\begin{equation*}
U(A)\gamma (x)=\gamma (A^*x)
\end{equation*}
This gives a group representation of $\gscript _4$ because
\begin{equation*}
U(AB)\gamma (x)=\gamma (B^*A^*x)=U(B)\gamma (A^*x)=U(A)U(B)\gamma (x)
\end{equation*}
so $U(AB)=U(A)U(B)$. An explicit expression for $U(A)\gamma (x)$ is
\begin{align*}
U(A)\gamma (x)&=\gamma (A^*x)=\sum _{p\in\Gamma _m}\frac{1}{p_0}
  \sqbrac{f(p)a(p)e^{i\pi pA^*x/2}+g(p)a(p)^*e^{-i\pi p A^*x/2}}\\
  &=\sum _{p\in\Gamma _m}\frac{1}{p_0}\sqbrac{f(p)a(p)e^{i\pi Apx/2}+g(p)a(p)^*e^{-i\pi Apx/2}}\\
  &\sum _{p\in\Gamma _m}\frac{1}{p_0}\sqbrac{f(A^*p)a(A^*p)e^{i\pi px/2}+g(A^*p)a(A^*p)^*e^{-i\pi px/2}}
\end{align*}

Suppose $\phi$ is the quantum field
\begin{equation*}
\phi (x)=\sum _{p\in\Gamma _m}\frac{1}{p_0}\sqbrac{a(p)e^{i\pi px/2}+a(p)^*e^{-i\pi px/2}}
\end{equation*}
We can form the quantum fields $\partial _\mu\phi (x)$, $\mu =0,1,2,3$. Thus,
\begin{align*}
\partial _0\phi (x)&\sum _{p\in\Gamma _m}\sqbrac{a(p)e^{i\pi px/2}-a(p)^* e^{i\pi px/2}}\\
  \partial _j\phi (x)&=-\sum _{p\in\Gamma _m}\frac{p_j}{p_0}\sqbrac{a(p)e^{i\pi px/2}-a(p)^*e^{-i\pi px/2}}
\end{align*}
It follows from the Klein-Gordon equation of Section~3 that
\begin{equation}        
\label{eq42}
\paren{\square -m^2}\phi (x)=0
\end{equation}
We call \eqref{eq42} the \textit{free-field equation} or the \textit{equation of motion}. For an interacting field, the equation of motion becomes
\begin{equation}        
\label{eq43}
\paren{\square -m^2}\phi (x)=j_\phi (x)
\end{equation}
where $j_\phi (x)$ is the \textit{current} for $\phi$. The current depends on the interaction experienced by the field. For example, suppose we have interacting bosons, one with mass $m$ as given before and another with mass $M$. The first has equation of motion \eqref{eq43} and the second has equation of motion
\begin{equation*}
\paren{\square -M^2}\sigma (x)=j_\sigma (x)
\end{equation*}
where $\sigma$ is the free quantum field
\begin{equation*}
\sigma (x)=\sum _{k\in\Gamma _M}\frac{1}{k_0}\sqbrac{a(k)e^{i\pi kx/2}+a(k)^*e^{-i\pi kx/2}}
\end{equation*}
and $\sqbrac{\phi (x),\sigma (y)}=0$ for $x,y\in\sscript _4$. The \textit{interaction Hamiltonian density} is gotten from the current by $\hscript (x)=j_\phi (x)\phi (x)$. For example a typical current is $j_\phi (x)=g\sigma (x)\phi (x)$ where $g$ is called the \textit{coupling constant}. We then have that \cite{vel94}
\begin{equation*}
\hscript (x)=g\sigma (x)\phi (x)^2
\end{equation*}
Now $\hscript$ should not depend on the particular equation of motion so we should also have that
$\hscript (x)=j_\sigma (x)\sigma (x)$. It follows that $j_\sigma (x)=g\phi (x)^2$. The equations of motion become
\begin{align*}
\paren{\square -m^2}\phi (x)&=g\sigma (x)\phi (x)\\
\paren{\square -M^2}\sigma (x)&=g\phi (x)^2
\end{align*}

Once we obtain the interaction Hamiltonian density $\hscript (x)$ we define the \textit{interaction Hamiltonian} $H(x_0)$ as follows. Letting
\begin{equation*}
V(x_0)=\ab{\brac{x\colon\doubleab{\bfx}_3\le x_0}}
\end{equation*}
be the cardinality of the set in brackets (called the \textit{space-volume}) we define
\begin{equation*}
H(x_0)=\frac{1}{V(x_0)}\sum\brac{\hscript (x_0,\bfx )\colon\doubleab{\bfx}_3\le x_0}
\end{equation*}
The self-adjoint operators $H(x_0)$ describe an interaction as a function of time $x_0=0,1,2,\ldots\,$. The main contact with observation is given by the corresponding \textit{scattering operators} $S(x_0)$ which satisfy the ``second quantization'' equation
\begin{equation}        
\label{eq44}
\nabla _{x_0}S(x_0)=iH(x_0)S(x_0)
\end{equation}
Of course, \eqref{eq44} is a generalization of Schr\"odinger's equation and in this discrete framework $\nabla _{x_0}$ is the difference operator
\begin{equation*}
\nabla _{x_0}S(x_0)=S(x_0+1)-S(x_0)
\end{equation*}
Starting with the initial condition $S(0)=I$ we obtain from \eqref{eq44} that
\begin{equation}        
\label{eq45}
S(n)=\sqbrac{I+iH(n-1)}\sqbrac{I+iH(n-2)}\cdots\sqbrac{I+iH(1)}\sqbrac{I+iH(0)}
\end{equation}
For $n\ne 0$, $S(n)$ is not unitary in general. However, presumably the \textit{limiting scattering operator}
$S=\lim\limits _{n\to\infty}S(n)$ should be unitary in order to preserve probability. The author has also considered another approach called reconditioning that essentially renormalizes at each step to maintain probability \cite{gud17}. If we multiply \eqref{eq45} out, we obtain the useful form
\begin{align}        
\label{eq46}
S(n)&=I+i\sum _{j=0}^{n-1}H(j)+i^2\sum _{j_2<j_1}^{n-1}H(j_1)H(j_2)
  +i^3\sum _{j_3<j_2<j_1}^{n-1}H(j_1)H(j_2)H(j_3)\notag\\
  &\quad +\cdots +i^nH(n-1)H(n-2)\cdots H(0)
\end{align}
Equation \eqref{eq46} gives a quantum inclusion-exclusion principle and is called a perturbation expansion.

In applications, the scattering operators are employed to find scattering amplitudes and probabilities. For example, suppose we have two particles with energy-momentum $p,q\in\Gamma _m$ and we seek the probability that after they interact, they scatter and attain energy-momentum, $p',q'\in\Gamma _m$. The initial and final scattering states in the Hilbert space $K$ are represented by unit vectors $\ket{pq}$ and $\ket{p'q'}$ in $K$, respectively. The \textit{final scattering amplitude} becomes
\begin{equation*}
\bra{p'q'}S\ket{pq}=\lim _{x_0\to\infty}\bra{p'q'}S(x_0)\ket{pq}
\end{equation*}
The corresponding probabilities are given by $\ab{\bra{p'q'}S(x_0)\ket{pq}}^2$ and $\ab{\bra{p'q'}S\ket{pq}}^2$.

We next consider spin $1/2$ fermion fields. Examples of spin $1/2$ particles are electrons, positrons, protons and neutrons. For definiteness, let us assume the fermion is an electron with mass $m$. We define the corresponding quantum field to be
\begin{align}        
\label{eq47}
\psi (x)&\sum _{p\in\Gamma _m}\frac{1}{p_0}\left[(a^1(p)u^1(p)+a^2(p)u^2(p))e^{i\pi px/2}\right.\notag\\
  &\qquad \left.+(b^1(p)^*u^3(p)+b^2(p)^*u^4(p))e^{-i\pi px/2}\right]
\end{align}
In \eqref{eq47}, $u^1(p),u^2(p),u^3(p),u^4(p)$ are the eigenvectors of $\mscript _p$ given by \eqref{eq31} and \eqref{eq32}. Moreover, $a^1(p),a^2(p),b^1(p),b^2(p)$ are annihilation operators for spin up electrons, spin down electrons, spin up positrons and spin down positrons, respectively. Of course, $a^1(p)^*,a^2(p)^*,b^1(p)^*,b^2(p)^*$
are the corresponding creation operators. Equation \eqref{eq47} is a four component operator whose components are
\begin{align*}
\psi _1(x)&=\sum _{p\in\Gamma _m}\frac{1}{p_0\sqrt{2p_o(p_0+m)}}
    \biggl[(p_0+m)a^1(p)e^{i\pi px/2}+p_3b^1(p)^*e^{-i\pi px/2}\biggr.\\
    &\hskip 12pc\biggl.+(p_1+ip_2)b^2(p)^*e^{-i\pi px/2}\biggr]\\
\psi _2(x)&=\sum _{p\in\Gamma _m}\frac{1}{p_0\sqrt{2p_o(p_0+m)}}
    \biggl[(p_0+m)a^2(p)e^{i\pi px/2}+(p_1+ip_2)b^1(p)^*e^{-i\pi px/2}\biggr.\\
    &\hskip 12pc\biggl.-p_3b^2(p)^*e^{-i\pi px/2}\biggr]\\
\psi _3(x)&=\sum _{p\in\Gamma _m}\frac{1}{p_0\sqrt{2p_o(p_0+m)}}
    \biggl[p_3a^1(p)e^{i\pi px/2}+(p_1-ip_2)a^2(p)e^{i\pi px/2}\biggr.\\
    &\hskip 12pc\biggl.+(p_0+m)b^1(p)^*e^{-i\pi px/2}\biggr]\\
\psi _4(x)&=\sum _{p\in\Gamma _m}\frac{1}{p_0\sqrt{2p_o(p_0+m)}}
    \biggl(p_1+ip_2)a^1(p)e^{i\pi px/2}-p_3a^2(p)e^{i\pi px/2}\biggr.\\
    &\hskip 12pc\biggl.+(p_0+m)b^2(p)^*e^{-i\pi px/2}\biggr]\\
\end{align*}
The field $\psi$ is not self-adjoint and we introduce the field
\begin{align*}
\psi (x)^*&=\sum _{p\in\Gamma _m}
    \biggl[\paren{a^1(p)^*\overu ^1(p)+a^2(p)^*\overu ^2(p)}e^{-i\pi px/2}\biggr.\\
    &\hskip 6pc\biggl.\paren{b^1(p)\overu ^3(p)+b^2(p)\overu ^4(p)}e^{i\pi px/2}\biggr]\\
\end{align*}
where $\overu ^j(p)$ are the complex-conjugates of $u^j(p)$, $j=1,2,3,4$. The components of $\psi (x)^*$ are
$\psi _j(x)^*$, $j=1,2,3,4$. We see that $\psi (x)$ annihilates electrons and creates positrons, while $\psi (x)^*$ creates electrons and annihilates positrons. We obtain a self-adjoint field by defining $\gamma (x)=\psi (x)+\psi (x)^*$.

Under the symmetry $A\in\gscript _4$, the fermion field transforms according to
\begin{align*}
U(A)\psi (x)&=A\psi (A^*x)\\
  &=\sum _{p\in\Gamma _m}\biggl[\paren{a^1(A^*p)Au^1(A^*p)+a^2(A^*p)Au^2(A^*p)}e^{i\pi px/2}\biggr.\\
    &\hskip 4pc\biggl.+\paren{b^1(A^*p)^*Au^3(A^*p)+b^2(A^*p)^*Au^4(A^*p)}e^{-i\pi px/2}\biggr]\\
\end{align*}
As before, $U$ gives a representation of $\gscript _4$ in the sense that $U(AB)=U(A)U(B)$.

We now introduce the \textit{Dirac operator}
\begin{align*}
\partial =\partial _0\gamma _0-\partial _1\gamma _1-\partial _2\gamma _2-\partial _3\gamma _3
  &=\begin{bmatrix}\noalign{\smallskip}
     \partial _0&0&-\partial _3&-\partial _1+i\partial _2\\\noalign{\smallskip}
     0&\partial _0&-\partial _1-i\partial _2&\partial _3\\\noalign{\smallskip}
     \partial _3&\partial _1-i\partial _2&-\partial _0&0\\\noalign{\smallskip}
     \partial _1+i\partial _2&-\partial _3&0&-\partial _0\\\noalign{\smallskip}\end{bmatrix}\\
\intertext{We have that}
\partial u^1(p)e^{i\pi px/2}&=\sqrt{\frac{p_0+m}{2p_0}}\,\partial
\begin{bmatrix}\noalign{\smallskip}
    e^{i\pi px/2}\\\noalign{\smallskip}
     0\\\noalign{\smallskip}
     \frac{p_3}{p_0+m}\,e^{i\pi px/2}\\\noalign{\medskip}
     \frac{p_1+ip_2}{p_0+m}\,e^{i\pi px/2}\\\noalign{\smallskip}\end{bmatrix}=mu^1(p)e^{i\pi px/2}\\
\intertext{In a similar way.}
     \partial u^2(p)e^{i\pi px/2}&=mu^2(p)e^{i\pi px/2}\\
      \partial u^3(p)e^{-i\pi px/2}&=mu^3(p)e^{-i\pi px/2}\\
     \partial u^4(p)e^{-i\pi px/2}&=mu^4(p)e^{-i\pi px/2}
\end{align*}
It follows that the \textit{Dirac equations} $(\partial -m)\psi (x)=0$ holds. The equation of motion for an interacting field would then have the form $(\partial -m)\psi (x)=j_\psi (x)$ where $j_\psi$ is a vector current. We also define the
\textit{conjugate Dirac operator}
\begin{equation*}
\overpartial =\partial _0\overgamma _0-\partial _1\overgamma _1-\partial _2\overgamma _2-\partial _3\overgamma _3
  =\begin{bmatrix}\noalign{\smallskip}
     \partial _0&0&-\partial _3&-\partial _1-i\partial _2\\\noalign{\smallskip}
     0&\partial _0&-\partial _1+i\partial _2&\partial _3\\\noalign{\smallskip}
     \partial _3&\partial _1+i\partial _2&-\partial _0&0\\\noalign{\smallskip}
     \partial _1-i\partial _2&-\partial _3&0&-\partial _0\\\noalign{\smallskip}\end{bmatrix}
\end{equation*}
and obtain the \textit{conjugate Dirac equation} $(\overpartial +m)\psi (x)^*=0$

We next consider an electromagnetic field. Let $e^j(k)$, $j=1,2,3$, be three mutually orthogonal 4-dimensional unit vectors satisfying
\begin{equation}        
\label{eq48}
\sum _{\mu =0}^3k_\mu e_\mu ^j(k)=0
\end{equation}
for every $k\in\Gamma _0$. The $k\in\Gamma _0$ describe (mass zero) photons and the $e^j(k)$, $j=1,2,3$ give the three polarization vectors for (spin~1) photons. The polarization vectors depend on the physical situation. An example is the following. When $k_2=k_3=0$, let
\begin{equation}        
\label{eq49}
e^1(k)=\frac{1}{\sqrt{k_0^2+k_1^2}}\,(k_1,-k_0,0,0)
\end{equation}
$e^2(k)=(0,0,1,0)$, $e^3(k)=(0,0,0,1)$. Otherwise, let $e^1(k)$ again be given by \eqref{eq49} and let
\begin{align*}
e^2(k)&\frac{1}{\sqrt{k_2^2+k_3^2}}\,(0,0,-k_3,k_2)\\
e^3(k)&=\frac{1}{N}\paren{k_0(k_2^2+k_3^2),k_1(k_2^2+k_3^2),-k_2(k_0^2+k_1^2),-k_3(k_0^2+k_1^2)}
\end{align*}
where the normalization constant $N=2k_0^2(k_0^2+k_1^2)(k_2^2+k_3^2)$. Corresponding to the polarization vectors $e^j(k)$ we have the three photon annihilation operators $a^j(k)$ and the three photon creation operators $a^j(k)^*$, $j=1,2,3$. The \textit{quantum electromagnetic field} $A_\mu (x)$, $\mu =0,1,2,3$, is given by
\begin{equation}        
\label{eq410}
A_\mu (x)=\sum _{k\in\Gamma _0}\frac{1}{k_0}\,\sum _{j=1}^3
  \sqbrac{e _\mu ^j(k)a^j(k)e^{i\pi px/2}+\overe _\mu ^j(k)a^j(k)^*e^{-i\pi px/2}}
\end{equation}
One can give a shorter notation for \eqref{eq410} by eliminating the subscript $\mu$. Notice that $A_\mu (x)$ is self-adjoint. Of course $A_\mu (x)$ is analogous to the electromagnetic vector potential in classical electromagnetism.

In quantum electrodynamics we have the free electron field $\psi$ with equation of motion $(\partial -m)\psi (x)=0$ and the free photon field $A$ with equation of motion $\square\,A_\mu (x)=0$. Applying \eqref{eq48} we obtain the
\textit{Lorentz condition}
\begin{equation*}
\sum _{\mu =0}^3\partial _\mu A_\mu (x)=0
\end{equation*}
The simplest nontrivial vector current for $\psi$ is
\begin{equation*}
j_\psi (x)=-e\sum _{\mu =0}^3A_\mu (x)\gamma _\mu\psi (x)
\end{equation*}
where $-e$ is the electron charge. This gives the interactive equation of motion
\begin{equation*}
(\partial -m)\psi (x)=-e\sum _{\mu =0}^3A_\mu (x)\gamma _\mu\psi (x)
\end{equation*}
The corresponding current for $A(x)$ is
\begin{equation*}
j_A(x)=-e\sum _{\mu =0}^3\psi (x)^*\gamma _\mu\psi (x)
\end{equation*}
and the interaction Hamiltonian density becomes
\begin{equation*}
\hscript (x)=-e\sum _{\mu =0}^3A_\mu (x)\psi (x)^*\gamma _\mu\psi (x)
\end{equation*}
We can now compute the interaction Hamiltonian $H(x_0)$ and the scattering operators $S(x_0)$ as previously discussed. This describes electron-electron scattering where the electrons interact by exchanging photons.

In a similar way, we can describe electron-proton scattering by photon exchange. The free electron field $\psi _e$ satisfies $(\partial -m)\psi _e(x)=0$ and the free proton field satisfies $(\partial -M)\psi _p(x)=0$ where $M$ is the proton mass. The interaction Hamiltonian density becomes
\begin{equation*}
\hscript (x)=e\sum _{\mu =0}^3A_\mu (x)
  \sqbrac{\psi _p(x)^*\gamma _\mu\psi _p(x)-\psi _e(x)^*\gamma _\mu\psi  _e(x)}
\end{equation*}

Finally, we mention the quantum fields for the weak spin~1 bosons $W^-$ and $W^+$ of mass $b$. Let $a^j(k)$ be the annihilation operators for $W^-$ and $b^j(k)$ the annihilation operators for $W^+$ with polarizations $j=1,2,3$. Then the quantum field for $W^-$ is given by
\begin{equation*}
W_\mu ^-(x)=\sum _{k\in\Gamma _b}\frac{1}{k_0}
   \sqbrac{\sum _{j=1}^3e_\mu ^j(k)a^j(k)e^{i\pi kx/2}+\overe _\mu ^j(k)b^j(k)^*e^{-i\pi kx/2}}
\end{equation*}
and $W_\mu ^+(x)=W_\mu ^-(x)^*$.


\begin{thebibliography}{99}
\bibitem{bdp16}A.~Bisco, G.~D'Ariano and P.~Perinotti, Special relativity in a discrete quantum universe,
arXiv: quant-ph 1503.01017v3 (2016).
\bibitem{cro16}D.~Crouse, On the nature of discrete space-time, arXiv: quant-ph 1608.08506v1 (2016).
\bibitem{gud16}S.~Gudder, Discrete scalar quantum field theory, arXiv: gen-ph 1610.07877v1 (2016).
\bibitem{gud17}S.~Gudder, Reconditioning in discrete quantum field theory, \textit{Int.\ J.\ Theor.\ Phys.}
DOI 10.1007/s10773-017-3350-6 (2017).
\bibitem{ps95}M.~Peskin and D.~Schroeder, \textit{An Introduction to Quantum Field Theory}, Addison-Wesely, Reading, Mass. (1995).
\bibitem{vel94}M.~Veltman, \textit{Diagrammatica}, Cambridge University Press, Cambridge (1994).


\end{thebibliography}
\end{document}